\documentclass[12pt,reqno]{amsart}
\usepackage{amsfonts}
\usepackage{amsmath}
\usepackage{amsthm}
\usepackage{mathrsfs}
\usepackage{amssymb}
\usepackage{stmaryrd}
\usepackage{graphicx}
\usepackage{epstopdf}
\usepackage{textcomp}
\usepackage[margin=1in]{geometry}

\newtheorem{theorem}{Theorem}

\theoremstyle{definition}

\newtheorem{remark}{Remark}

\newtheorem{conjecture}{Conjecture}
\begin{document}
\title[Asymptotic efficiency of the proportional compensation scheme]{Asymptotic efficiency of the proportional compensation scheme for a large number of producers}

\author{Dmitry B. Rokhlin}
\author{Anatoly Usov}

\address{Institute of Mathematics, Mechanics and Computer Sciences,
              Southern Federal University,
Mil'chakova str., 8a, 344090, Rostov-on-Don, Russia}
\email[Dmitry B. Rokhlin]{rokhlin@math.rsu.ru}
\email[Anatoly Usov]{usov@math.rsu.ru}


\begin{abstract}
We consider a manager, who allocates some fixed total payment amount between $N$ rational agents in order to maximize the aggregate production. The profit of $i$-th agent is the difference between the compensation (reward) obtained from the manager and the production cost. We compare (i) the \emph{normative} compensation scheme, where the manager enforces the agents to follow an optimal  cooperative strategy;  (ii) the \emph{linear piece rates} compensation scheme, 
where the manager announces an optimal reward per unit good; (iii) the \emph{proportional} compensation scheme, where agent's reward is proportional to his contribution to the total output. Denoting the correspondent total production levels by $s^*$, $\hat s$ and $\overline s$ respectively, where the last one is related to the unique Nash equilibrium, we examine the limits of the prices of anarchy $\mathscr A_N=s^*/\overline s$, $\mathscr A_N'=\hat s/\overline s$ as $N\to\infty$. These limits are calculated for the cases of identical convex costs with power asymptotics at the origin, and for power costs, 
corresponding to the Coob-Douglas and generalized CES production functions with decreasing returns to scale. Our results show that asymptotically no performance is lost in terms of $\mathscr A'_N$, and in terms of $\mathscr A_N$ the loss does not exceed $31\%$.
\end{abstract}

\maketitle

\section{Introduction} \label{sec:1}
\setcounter{equation}{0}
Consider a manager, who allocates some fixed some fixed total payment amount $M$ between $N$ producers (agents) in order to maximize the aggregate production. The profit of each agent equals to the difference between the reward obtained from the manager and the production cost $\varphi_i$. If the cost functions are known, the manager can determine rewards, stimulating the optimal aggregate production $s^*$. We call such compensation scheme \emph{normative}. Besides the quite unrealistic assumption that the cost functions are known, this scheme suffers from another drawback: it does not announce any common reward sharing rules. However, $s^*$ can serve as a benchmark.

Another idea is to use the \emph{linear piece rates} compensation scheme, announcing a price $\mu$ of the unit good. So, the reward $\mu x_i$ of $i$-th agent will be linear in his production level $x_i$. Assuming an individually optimal (rational) agent behaviour, the manager can chose $\mu$ in such a way that the total reward does not exceed $M$, and the total production $\hat s$ cannot be improved by another linear reward rule. Clearly, this compensation scheme also requires the knowledge of production cost functions, although it is easier to assign one parameter $\mu$, rather than the full set of rewards, as in the normative scheme. Note also that in the piece rates allocation scheme the total reward requested by \emph{irrational} agents can exceed $M$. Nevertheless, we regard the value $\hat s$ as another benchmark.

The main focus of the present study is the \emph{proportional} compensation scheme, where the reward $M x_i/(x_1+\dots+x_N)$ of $i$-th agent is proportional to his contribution to the aggregate production. The realization of this scheme requires no information concerning the cost functions, and the total reward equals to $M$ irrespective of agent actions (except the trivial case, where $x=0$). So, the manager allows the agents to determine optimal production levels on their own in the course of a (non-cooperative) game with the payoff functions
\begin{equation} \label{1.1}
 M\frac{x_i}{x_1+\dots+x_N}-\varphi_i(x_i).
\end{equation} 
Under the assumption that the cost functions $\varphi_i$ are convex and strictly increasing, the game (\ref{1.1}) has a unique Nash equilibrium. By $\overline s$ we denote the correspondent total production. 
 
One may argue that the computation of a Nash equilibrium also requires the knowledge of cost functions. However, such equilibrium also can emerge as a result of agent interaction in a repeated game through the mechanism of no-regret learning. We recall this concept at the end of the paper. For each agent the no-regret learning does not require the knowledge of the cost functions of other agents.

The game (\ref{1.1}) is a special case of the Cournot oligopoly: \cite{OkuSzi99,MouQua16}, and it fits into the extensively studied theory of \emph{contests}: see \cite{Cor07,Kon09,CorSer16,Voj16} for reviews (an experimental research is reviewed in \cite{DecKovShe15}). 
In a contest the payoff function of each player is the difference between the \emph{contest success function} (CSF)  and the cost of player's effort. A player's CSF usually equals to the expected value of winning an indivisible prize, or, as in our case, to the portion of the prize, obtained by the player. It depends on the efforts of all players, and it is increasing in the effort of a selected player and decreasing in the efforts of the other ones. An account of the CSF's can be found in \cite{JiaSkaVai13}. 

Using the substitution $f_i(y_i)=\varphi_i^{-1}(y_i)$, we can reduce the game (\ref{1.1}) to a strategically equivalent contest 
\begin{equation} \label{1.2}
 M\frac{f_i(y_i)}{f_1(y_1)+\dots+f_N(y_N)}-y_i
\end{equation}  
with the CSF of the general-logit form (in the terminology of \cite[Chapter 4]{Voj16}). In addition, in our main example of power costs $\varphi_i=c_i x_i^\alpha$, $\alpha\ge 1$, corresponding to generalized CES production functions with decreasing returns to scale, the game (\ref{1.2}) boils down to the \emph{Tullock contest} with $f_i(y_i)=(y_i/c_i)^{1/\alpha}$. 
 
Contests are used to model conflict situations in rent-seeking, resource allocation, patent races, sports, advertising, etc.  
The present paper is related to the analysis  of relative performance incentive schemes in labour contracts. An active study of such problems was initiated in 1980s: \cite{LazRos81,GreSto83,NalSti83,Mal86}. 
We also mention several recent papers with an emphasis on experimental and empirical studies:
\cite{Con14,Guth16,She16}, where the reader can find a lot of additional references. 
The prevailing concept is the rank-order allocation of prizes (rank-order tournaments). However, the proportional prize-contest was promoted by the means of experimental studies in \cite{CasMasShe10}.

The main feature of the present paper is the analysis of the following two versions of the ``price of anarchy'': 
$\mathscr A_N=s^*/\overline s,$ $\mathscr A_N'=\hat s/\overline s$
for a large number $N$ of agents. In line with \cite{KouPap09} the price of anarchy shows how much performance is lost by the lack of coordination. The study of the prices of anarchy recently became an active area of research. We mention only a few papers, studying an efficiency of the proportional resource allocation mechanism in somewhat different models: \cite{JohTsi04,ChriSgoTan16,CarVou16}. 

In Section \ref{sec:2} we describe three compensation schemes mentioned above. In particular, we point out that any contest scheme cannot be better than the normative one (Remark \ref{rem:4}). In Section \ref{sec:3} we study the prices of anarchy 
$\mathscr A_N,$ $\mathscr A_N'$ for large $N$. Our results show that for the cases of identical convex costs $\varphi_i=\varphi$ with power asymptotics at the origin (Theorem \ref{th:1}), and for heterogeneous agents with power costs $\varphi_i(x)=c_i x^\alpha$, $\alpha> 1$ (Theorem \ref{th:2}), asymptotically no performance is lost in terms of $\mathscr A'_N$, and in terms of $\mathscr A_N$ the loss does not exceed $31\%$. These results characterize an asymptotic efficiency of the proportional compensation scheme. We also conjecture that this result remains true for heterogeneous agents with linear cost functions ($\alpha=1$) and i.i.d. marginal costs $c_i$.

\section{Three compensation schemes} \label{sec:2}
\setcounter{equation}{0}
 Let $x_i$ be the amount of good produced by $i$-th agent. Denote by  $\varphi_i:\mathbb R_+\mapsto\mathbb R_+$ the related production cost. We assume that the functions $\varphi_i$ are twice continuously differentiable, $\varphi_i(0)=0$, $\varphi_i'(x_i)>0$, $\varphi_i''(x_i)\ge 0$, $x_i>0$. It easily follows that $\varphi_i(x_i)\to+\infty$, $x_i\to+\infty$.

(i) \emph{Normative compensation scheme}. Agent $i$ knows the reward function $\psi_i(x_i)\ge 0$ at the beginning of the production cycle and maximizes his profit: 
\begin{equation} \label{2.1}
 \psi_i(x_i)-\varphi(x_i)\to\max_{x_i\ge 0}.
\end{equation} 
Let $\psi_i(0)=0$, and denote by $\tilde x_i=\tilde x_i(\psi_i)$ optimal solutions of (\ref{2.1}), which for simplicity we assume to exist.
The manager has $M$ units of capital at his disposal. His aim is to maximize the total production:
$$ \sum_{i=1}^N \tilde x_i\to\max$$
over all reward functions $\psi_i$, satisfying the conditions
$$\sum_{i=1}^N\psi_i(\tilde x_i)\le M;\quad \psi_i\ge 0,\ i=1,\dots,N.$$

 Since $\psi_i(\tilde x_i)-\varphi_i(\tilde x_i)\ge 0$, we get the estimate
$$ \sum_{i=1}^N \varphi_i(\tilde x_i)\le\sum_{i=1}^N \psi_i(\tilde x_i)\le M.$$
Thus, given the budget $M$, the total production cannot exceed the value 
\begin{equation} \label{2.2}
 s^*=\sup\left\{\sum_{i=1}^N x_i:\sum_{i=1}^N \varphi_i(x_i)\le M,\quad x\ge 0\right\}
\end{equation}
for any kind of rewards $\psi_i$. 

On the other hand, it is possible to obtain the total production arbitrary close to $s^*$ by announcing the rewards 
\begin{equation} \label{2.3}
\psi_i(x_i)=\varphi_i(x_i^*) I_{[x_i^*-\varepsilon_i,\infty)}(x_i),
\end{equation}
where $x^*=(x_i^*)_{i=1}^N$ is an optimal solution of (\ref{2.2}) and
$$\varepsilon_i\in (0,x_i^*), \quad\textrm{if }\ x_i^*>0;\quad \varepsilon_i=0,\quad\textrm{if }\ x_i^*=0.$$
Indeed, in this case the optimal solution of (\ref{2.1}) is of the form 
$$\tilde x_i=\begin{cases}
x_i^*-\varepsilon_i, & x_i^*>0,\\
0,& x_i^*=0
\end{cases}
$$ 
and $\sum_{i=1}^N \tilde x_i=s^*-\sum_{i=1}^N \varepsilon_i$, while
$\sum_{i=1}^N\psi_i(\tilde x_i)=\sum_{i=1}^N\varphi_i(x_i^*)\le M.$
Thus, one can regard $s^*$ as the optimal total production amount under the normative compensation scheme.

(ii) \emph{Linear piece rates compensation scheme}. Assume that the cost functions are strictly convex and the manager tries to choose a best linear reward function $\psi_i(x_i)=\mu x_i$. The production levels $\hat x_i(\mu)$ are determined by the problems
\begin{equation} \label{2.4}
 \mu x_i-\varphi_i(x_i)\to\max_{x_i\ge 0}.
\end{equation} 
The functions $\hat x_i(\mu)$ are non-decreasing, and the best choice of $\mu$ corresponds to the largest total production which does not violate the budget constraint:
\begin{equation} \label{2.5}
\mu\sum_{i=1}^N\hat x_i=M.
\end{equation} 
The aggregate production is given by
$\hat s=\sum_{i=1}^N\hat x_i.$

(iii) \emph{Proportional compensation scheme}. The left-hand side of (\ref{2.5}) equals to the total reward. If the agents anticipate that the manager selects $\mu$ in this way, then they become involved in the non-cooperative game with the payoff functions
\begin{equation} \label{2.6}
  H_i(x)=M\frac{x_i}{\sum_{i=j}^N x_j}-\varphi_i(x_i),\quad x\ge 0
\end{equation} 
with the convention $0/0=0$. The total production equals to
$$ \overline s=\sum_{i=1}^N \overline x_i,$$
where $\overline x$ is the unique Nash equilibrium (in pure strategies) of the game (\ref{2.6}):
$$ H_i(\overline x_1,\dots,\overline x_i,\dots,\overline x_N)\ge H_i(\overline x_1,\dots,x_i,\dots,\overline x_N),\quad j=1,\dots,N,\quad x_i\ge 0.$$
The existence and uniqueness of a Nash equilibrium (for $N\ge 2$) was proved in \cite{SziOku97}. The proof was simplified in \cite{CorHar05}, see also \cite{Cortak16} for an exposition. 

It is easy to see that $\overline x$ has at least two positive components. Furthermore, for such $\overline x$ the functions $x_i\mapsto H_i(\overline x_1,\dots,x_i,\dots,\overline x_N)$ are (strictly) concave. 
An elementary analysis of the correspondent one-dimensional problems shows that $\overline x$ is characterized by the following relations  
 \begin{equation} \label{2.6A}
 \varphi_i'(\overline x_i)=M\frac{\overline s -\overline x_i}{\overline s^2},\quad \textrm{if }\ \varphi_i'(0)<\frac{M}{\overline s},
 \end{equation}
  \begin{equation} \label{2.6B}
 \overline x_i=0,\quad \textrm{if }\ \varphi_i'(0)\ge\frac{M}{\overline s},
  \end{equation}
    \begin{equation} \label{2.6C}
\overline s= \sum_{j=1}^N\overline x_j.
  \end{equation}
Following \cite{SziOku97}, note that for $s>0$ the relations
\begin{equation} \label{2.6D}
 s^2 \varphi_i'(\overline z_i)=M(s-\overline z_i),\quad s\varphi_i'(0)<M,
\end{equation} 
$$ \overline z_i(s)=0,\quad s\varphi_i'(0)\ge M$$
uniquely define continuous functions $\overline z_i(s)$.
Clearly, $\overline x$ is a Nash equilibrium iff $\overline x=\overline z(\overline s)$, where $\overline s$ is a solution of the equation
\begin{equation} \label{2.6F}
 \sum_{i=1}^N\overline z_i(s)=s,\quad s>0.
\end{equation}  

Following \cite{CorHar05,Cortak16} let us pass from the \emph{replacement functions} $\overline z_i$ to the \emph{share functions} $\overline\sigma_i(s)=\overline z_i(s)/s$. The existence and uniqueness of a Nash equilibrium follow from (\ref{2.6F}) in view of the properties of the share functions (see \cite[Proposition 2]{Cortak16}): $\overline\sigma_i$ are continuous, strictly decreasing where positive,
$$ \lim_{s\to 0}\overline\sigma_i(s)=1,\quad \lim_{s\to\infty}\overline\sigma_i(s)=0.$$
So, the equation
\begin{equation} \label{2.6G}
\sum_{i=1}^N\overline\sigma_i(s)=1,\quad s>0,
\end{equation}
which is equivalent to (\ref{2.6F}), has a unique solution.

\begin{remark} \label{rem:1} 
Introducing the game (\ref{2.6}), we followed the reasoning of \cite{JohTsi04}. In their model (inspired by \cite{Kel97}),  the users share a communication link of some given capacity. The link manager gets payments (bids) from the users and allocates the rates according to the announced price. The manager adjusts the price in order to allocate the entire link capacity. If the users are price takers, then the model is referred to as a competitive equilibrium. If they are price-anticipating, then they are involved in a game, and it is assumed that their bids correspond to a Nash equilibrium.
\end{remark} 

The reward functions (\ref{2.3}), in fact, only tell the agents the production levels $x_i^*-\varepsilon_i$, specified for them by the manager. This normative scheme is quite sensible to individual cost functions $\varphi_i$. Moreover, it does not announce any common compensation rules. All these drawbacks force to seek for more reliable compensation schemes. The following discussion shows that this task is not trivial. 

We see that (\ref{2.2}) is a solvable convex optimization problem, satisfying the Slater condition. Hence, $x^*\ge 0$ is an optimal solution of (\ref{2.2}) iff there exists $\lambda^*\ge 0$ such that 
\begin{equation} \label{2.7}
\lambda^*\varphi_j'(x_j^*)=1,\quad \textrm{if }\ x_j^*>0;\qquad
\lambda^*\varphi_j'(0)\ge 1,\quad \textrm{if }\ x_j^*=0;
\end{equation}
$$ \lambda^*\left(\sum_{i=1}^N\varphi_i(x_i^*)-M\right)=0,\quad \sum_{i=1}^N\varphi_i(x_i^*)\le M.$$
Equivalently, $x^*\ge 0$ is an optimal solution of (\ref{2.2}) iff there exists $\lambda^*>0$ such that (\ref{2.7}) and the equality
\begin{equation} \label{2.8}
\sum_{i=1}^N\varphi_i(x_i^*)=M
\end{equation}
hold true. Furthermore, for given $\lambda^*>0$ a point $x^*\ge 0$ satisfies (\ref{2.7}) iff each $x_i^*$ is an optimal solution of the problem
\begin{equation} \label{2.9}
 x_i/\lambda^*-\varphi_i(x_i)\to\max_{x_i\ge 0}
\end{equation}
similar to (\ref{2.4}). Assume for a moment that $\varphi_i$ are strictly convex. Then (\ref{2.8}) implies that $x^*$ is unique and $\lambda^*$ is also uniquely defined by (\ref{2.7}), since at least one component of $x^*$ is positive. 

It is tempting to try $\psi_i(x_i)=x_i/\lambda^*$ for the role of reward functions. Indeed, by (\ref{2.9}), they stimulate optimal production levels $x_i^*$. However, in contrast to the piece rate scheme, $\psi_i(x_i)=x_i/\lambda^*$ are not legal reward functions, since  $x_i^*/\lambda^*-\varphi_i(x_i^*)>0$ for $x_i^*>0$, and the total reward exceeds the budget $M$:
$$ \sum_{i=1}^N\psi_i(x_i^*)=\frac{1}{\lambda^*}\sum_{i=1}^N x_i^*>\sum_{i=1}^N\varphi_i(x_i^*)=M.$$

Note, that the substitution $x_i=\varphi_i^{-1}(y_i)$ reduces (\ref{2.2}) to the following equivalent problem:
\begin{equation} \label{2.10}
 \sup\left\{\sum_{i=1}^N U_i(y_i):\sum_{i=1}^N y_i\le M,\quad y\ge 0\right\},
\end{equation}
where $U_i(y_i)=\varphi_i^{-1}(y_i)$ are strictly increasing concave functions. This is a customary nonlinear resource allocation problem: see, e.g., \cite{Pat08}. If the functions $U_i$ are strictly concave, then, similarly to the above discussion, there is a unique pair $(y^*,\mu^*)$ with $y^*\ge 0$, $\mu^*>0$, satisfying the optimality conditions 
$$ U_i'(y_i^*)=\mu^*,\quad \textrm{if }\ y_i^*>0;\qquad U_i'(0)\le\mu^*,
\quad \textrm{if }\ y_i^*=0;\qquad \sum_{i=1}^N y_i^*=M.$$
It follows that the unique optimal solution $y^*$ of (\ref{2.10}) can be recovered from the one-dimensional optimization problems
\begin{equation} \label{2.11}
 U_i(y_i)-\mu^* y_i\to\max_{y_i\ge 0}.
\end{equation}
Thus, by selling the resource at price $\mu^*$ (per unit), the manager can stimulate the optimal plan $y^*$.
But in the present context $y_i=\varphi_i(x_i)$ correspond to production costs, so the optimization problems (\ref{2.11}) make no economic sense.

\begin{remark} \label{rem:4}
Closing this section, we will show that any contest scheme cannot produce better result than (\ref{2.2}). Consider a non-cooperative game between $N$ agents with the payoff functions
$$H_i(x)=\Psi_i(x_1,\dots,x_N)-\varphi_i(x),\quad x\ge 0,$$
where $\Psi_i\ge 0$ is the reward of $i$-th agent, and $\Psi_i(0)=0$. Let a random vector $(\overline\xi_1,\dots,\overline\xi_N)\ge 0$ be a Nash equilibrium (in mixed strategies):
$$\mathsf E(\Psi_i(\overline\xi_1,\dots,\overline\xi_n)-\varphi_i(\overline\xi_i))\ge \mathsf E(\Psi_i(\overline\xi_1,\dots,\xi_i,\dots,\overline\xi_n)-\varphi_i(\xi_i)),\quad \xi_i\ge 0.$$
We implicitly assume that all expectations exist. Putting $\xi_i=0$, we infer that 
$$\mathsf E(\Psi_i(\overline\xi_1,\dots,\overline\xi_n)-\varphi_i(\overline\xi_i))\ge 0.$$
If the total reward on average does not exceed $M$: $\sum_{i=1}^N \mathsf E \Psi_i(\overline\xi_1,\dots,\overline\xi_n)\le M$, then
$$ \sum_{i=1}^N \mathsf E \varphi_i(\overline\xi_i)\le M.$$
A fortiori, $\sum_{i=1}^N  \varphi_i(\mathsf E\overline\xi_i)\le M$ by the Jensen inequality, and from the definition (\ref{2.2}) of $s^*$ it follows that
$$ \sum_{i=1}^N  \mathsf E\overline\xi_i\le s^*.$$

This negative result is by no means an indication that contest compensation schemes are useless. The point is that 
the organization of a contest may not require the knowledge of production cost functions $\varphi_i$. 
\end{remark}

\section{The prices of anarchy in case of a large number of producers} \label{sec:3}
\setcounter{equation}{0}
In this paper we are interested in the behaviour of the following two versions of the ``price of anarchy'': 
$$\mathscr A_N=s^*/\overline s,\qquad \mathscr A_N'=\hat s/\overline s$$
for a large number $N$ of agents. 
Recall, that the quantities $s^*$, $\hat s$, $\overline s$ describe three types of the aggregate production: 
\begin{itemize}
\item[(i)] $s^*$ corresponds to an optimal cooperative strategy, enforced by the reward functions (\ref{2.3}): 
see (\ref{2.2}) (the normative compensation scheme);
\item[(ii)] $\hat s$ is related to the case of ``reward-taking'' agents: see (\ref{2.4}), (\ref{2.5}), where the ``best'' common linear reward function is announced by the manager (the linear piece rates compensation scheme);  
\item[(iii)] $\overline s$ corresponds to the Nash equilibrium of the game (\ref{2.6}) for ``reward-an\-ti\-ci\-pa\-ting'' agents (the proportional compensation scheme). 
\end{itemize}
We will refer to the related problems as (i), (ii) and (iii). 

The case of identical cost functions is considered in the following theorem.
\begin{theorem} \label{th:1}
Assume that $\varphi_i=\varphi$ and
\begin{equation} \label{3.1}
\varphi(y)\sim c y^\alpha,\quad \varphi'(y)\sim \alpha c y^{\alpha-1},\quad y\to +0;\quad c>0,\quad \alpha\ge 1.
\end{equation} 
Then 
$$\lim_{N\to\infty}\mathscr A_N=\alpha^{1/\alpha},\quad \lim_{N\to\infty}\mathscr A_N'=1.$$
 \end{theorem}
\begin{proof} (i) As it was mentioned in Section \ref{sec:2}, a vector $x^*$ is an optimal solution of (\ref{2.2}) iff it satisfies (\ref{2.8}), and there exists $\lambda^*\ge 0$, satisfying (\ref{2.7}). It is natural to seek a solution of the (\ref{2.7}), (\ref{2.8}) in the symmetric form: $x_i^*=y^*>0$, $i=1,\dots,N$. We have 
\begin{equation} \label{3.2}
 \lambda^*\varphi'(y^*)=1,\quad \lambda^*\ge 0;\quad N\varphi(y^*)=M.
\end{equation}  
Clearly, such a pair $(y^*,\lambda^*)$ exists. From the second equality (\ref{3.2}) it follows that $y^*\to 0$, $N\to\infty$, and using the first condition (\ref{3.1}), we get 
\begin{equation} \label{3.3}
y^*\sim\left(\frac{M}{cN}\right)^{1/\alpha},\quad s^*=N y^*\sim\left(\frac{M}{c}\right)^{1/\alpha} N^{(\alpha-1)/\alpha},\quad N\to\infty.
\end{equation} 

(ii) From (\ref{2.4}), (\ref{2.5}) we see that $\hat x_i$ are identical: $\hat x_i(\mu)=\hat y$, and
\begin{equation} \label{3.4}
\varphi'(\hat y)=\mu,\quad N\mu \hat y=M.
\end{equation} 
Using (\ref{3.4}) and (\ref{3.1}), we conclude that $\hat y\to 0$, $N\to\infty$ and
$$\mu \hat y=\frac{M}{N}=\hat y\varphi'(\hat y)\sim \alpha c \hat y^{\alpha};\quad \hat y\sim\left(\frac{M}{\alpha cN}\right)^{1/\alpha},$$
\begin{equation} \label{3.5}
\hat s=N \hat y\sim\left(\frac{M}{\alpha c}\right)^{1/\alpha}N^{(\alpha-1)/\alpha},\quad N\to\infty.
\end{equation} 

(iii) We look for a symmetric Nash equilibrium of (\ref{2.6}): $\overline x_i=\overline y>0$, $i=1,\dots,N$. From (\ref{2.6A}), (\ref{2.6C}) we get
$$ \varphi'(\overline y)=M\frac{N-1}{N^2 \overline y}.$$
Hence, $\overline y\to 0$, $N\to\infty$ and
$$ M\frac{N-1}{N^2}=\overline y\varphi'(\overline y)\sim \alpha c \overline y^{\alpha};\quad
\overline y\sim \left(\frac{M}{\alpha cN}\right)^{1/\alpha},
$$
\begin{equation} \label{3.6}
\overline s=N \overline y\sim \left(\frac{M}{c\alpha}\right)^{1/\alpha} N^{(\alpha-1)/\alpha},\quad N\to\infty.
\end{equation} 

The assertion of the theorem follows from the asymptotic forms (\ref{3.3}), (\ref{3.5}), (\ref{3.6}). 
\end{proof}

Assume that the agents use the same technology, but obtain resources at different prices. This situation is natural if the firm has departments in various locations. In this case the resource prices may depend on the quality of transportation network, the cost of labour, etc., in a concrete location. Denote by $(r_1,\dots,r_m)$ the resource amounts (inputs), and by $(p_1^i,\dots,p_m^i)$ their prices in $i$-th location.  For the production function $F(r_1,\dots,r_m)$ the production cost function is defined by
$$ \varphi_i(x)=\inf\left\{\sum_{j=1}^m p_j^i r_j: F(r_1,\dots,r_m)\ge x,\ r\ge 0\right\}.$$

For the Cobb-Douglas production function $F(r)=A\prod_{j=1}^m r_j^{\beta_j}$, $A>0$, $\beta_j>0$ by the Lagrange duality (see, e.g., \cite{BoyVan04}) we have:
$$ \varphi_i(x)=\inf\left\{\sum_{j=1}^m p_j^i r_j: \sum_{j=1}^m \beta_j \ln r_j\ge\ln (x/A),\ r\ge 0\right\}=\sup_{\lambda\ge 0}\theta_i(\lambda),$$
where $\ln 0=-\infty$ and
\begin{align*}
\theta_i (\lambda)&=\inf_{r\ge 0}\left\{\sum_{j=1}^m p_j^i r_j +\lambda\left(\ln (x/A)-\sum_{j=1}^m \beta_j \ln r_j\right)\right\}\\
&=\lambda\ln (x/A)+\sum_{j=1}^m \inf_{r_j\ge 0}\{p^i_j r_j-\lambda\beta_j\ln r_j\}\\
&=\lambda\ln (x/A)+\sum_{j=1}^m \left(\lambda\beta_j-\lambda\beta_j\ln\frac{\lambda\beta_j}{p_j^i}\right).
\end{align*}
An elementary calculation shows that
$$ \varphi_i(x)=\sup_{\lambda\ge 0}\theta_i(\lambda)=c_i x^\alpha,\quad
 \alpha=\frac{1}{\sum_{j=1}^m\beta_j},\quad c_i=\frac{1}{\alpha A^\alpha}\left(\prod_{j=1}^m \left(\frac{p_j^i}{\beta_j}\right)^{\beta_j}\right)^\alpha.$$

Similarly, for the generalized CES production function (see, e.g., \cite{Che12,VilVil17}): 
$$F(r)=A\left(\sum_{j=1}^m a_j^\rho r_j^\rho\right)^{\gamma/\rho},\quad A,a_j,\gamma>0,\quad \rho\in (0,1)$$
we have
$$ \varphi_i(x)=\inf\left\{\sum_{j=1}^m p_j^i r_j: \sum_{j=1}^m a_j^\rho r_j^\rho\ge\left(\frac{x}{A}\right)^{\rho/\gamma},\ r\ge 0\right\}=\sup_{\lambda\ge 0}\theta_i(\lambda),$$
where
\begin{align*}
\theta_i(\lambda)&=\inf_{r\ge 0}\left\{\sum_{j=1}^m p_j^i r_j+\lambda\left(\left(\frac{x}{A}\right)^{\rho/\gamma}-\sum_{j=1}^m a_j^\rho r_j^\rho\right) \right\}\\
&=\lambda\left(\frac{x}{A}\right)^{\rho/\gamma}+\sum_{j=1}^m\inf_{r_j\ge 0}\{p_j^i r_j- \lambda a_j^\rho r_j^\rho\}\\
& =\lambda\left(\frac{x}{A}\right)^{\rho/\gamma}-(1-\rho)\sum_{j=1}^m \left(\frac{a_j\rho}{p_j^i}\right)^{\rho/(1-\rho)}\lambda^{1/(1-\rho)}.
\end{align*}
Maximizing this expression over $\lambda\ge 0$, we get
$$ \varphi_i(x)=c_i x^{1/\gamma},\quad c_i=\frac{1}{A^{1/\gamma}}\left(\sum_{j=1}^m\left(\frac{a_j}{p_j^i}\right)^{\rho/(1-\rho)}\right)^{-(1-\rho)/\rho}.$$

Thus, the Cobb-Douglas and generalized CES production functions with decreasing returns to scale (that is, with $\sum_{j=1}^m\beta_j\le 1$ and $\gamma\le 1$ respectively) correspond to power cost functions:
$\varphi_i(x)=c_i x^{\alpha}$, $\alpha\ge 1.$
Certainly, this fact is known (see \cite[Chapter 5]{CotMil99}), and we only recalled it here. 

Now we have enough economic motivation to consider a model, representing heterogeneous agents by power cost functions with common exponent and different multiplication constants.
\begin{theorem} \label{th:2}
Assume that $\varphi_i(x)=c_i x^\alpha$, $\alpha>1$, $c_i>0$ and
\begin{equation} \label{3.7}
 \lim_{N\to\infty}\sum_{i=1}^N \left(\frac{\min_{1\le k\le N} c_k}{c_i}\right)^{1/(\alpha-1)}=\infty.
\end{equation} 
Then 
\begin{equation} \label{3.7A}
\lim_{N\to\infty}\mathscr A_N=\alpha^{1/\alpha},\quad \lim_{N\to\infty}\mathscr A_N'=1.
\end{equation}
\end{theorem}
\begin{proof} (i) By the Lagrange duality the value of the problem (\ref{2.2}) can be represented as follows:
$$ s^*=-\inf\left\{-\sum_{i=1}^N x_i:\sum_{i=1}^N c_i x_i^\alpha\le M,\ x\ge 0\right\}=-\sup_{\lambda\ge 0}\theta(\lambda),$$
\begin{align*}
\theta(\lambda) &=\inf_{x\ge 0}\left\{-\sum_{i=1}^N x_i+\lambda\left(\sum_{i=1}^N c_i x_i^\alpha- M\right)\right\}=
-M\lambda+\sum_{i=1}^N \inf_{x_i\ge 0}\left(-x_i+\lambda c_i x_i^\alpha \right)\\
&=-M\lambda-B \lambda^{-\frac{1}{\alpha-1}},\quad
B=(\alpha-1)\left(\frac{1}{\alpha}\right)^{\frac{\alpha}{\alpha-1}}\sum_{i=1}^N\left(\frac{1}{c_i}\right)^{\frac{1}{\alpha-1}}.
\end{align*}
Maximizing this expression over $\lambda\ge 0$, we get
\begin{equation} \label{3.8}
 s^*=M^{1/\alpha}\left(\sum_{i=1}^N\frac{1}{c_i^{1/(\alpha-1)}}\right)^{(\alpha-1)/\alpha}.
\end{equation} 

(ii) The optimization problems (\ref{2.4}) take the form
$$ \mu x_i-c_i x_i^\alpha\to\max_{x_i\ge 0}.$$
Substituting their optimal solutions $\hat x_i=\left(\mu/(c_i\alpha)\right)^{1/(\alpha-1)}$ 
in (\ref{2.5}), we get
$$ \mu=\alpha^{1/\alpha}\left(\frac{M}{\sum_{i=1}^N c_i^{-1/(\alpha-1)}}\right)^{(\alpha-1)/\alpha}.$$
Hence,
\begin{align} \label{3.9}
\hat s &=\sum_{i=1}^N\hat x_i=\left(\frac{\mu}{\alpha}\right)^{1/(\alpha-1)}\sum_{i=1}^N\frac{1}{c_i^{1/(\alpha-1)}}\nonumber\\
&=\left(\frac{M}{\alpha}\right)^{1/\alpha}\left(\sum_{j=1}^N\frac{1}{c_j^{1/(\alpha-1)}}\right)^{(\alpha-1)/\alpha}
\end{align}
Comparing with (\ref{3.8}), we see that $\hat s=s^*/\alpha^{1/\alpha}.$

(iii) To analyse the proportional compensation scheme consider the equations (\ref{2.6D}), (\ref{2.6F}):
\begin{equation} \label{3.10}
 \chi_i(s, z_i)=s^2\alpha c_i z_i^{\alpha-1}-Ms+Mz_i=0,
\end{equation} 
$$ \chi(s,z)=\sum_{i=1}^N z_i-s=0.$$
For
\begin{equation} \label{3.10A}
 \hat z_i(s)=\left(\frac{M}{\alpha c_i s}\right)^{1/(\alpha-1)},\quad K\in (0,1)
\end{equation} 
we have $\chi_i(s, \hat z_i(s))=M\hat z_i(s)>0$ and
\begin{align*}
\chi_i(s, K\hat z_i(s)) &=M(K^{\alpha-1}-1)s+M K\hat z_i(s)<0,\quad \textrm{for }\ s>s_i,\\
s_i &=\left(\frac{K}{1-K^{\alpha-1}}\right)^{(\alpha-1)/\alpha} \left(\frac{M}{\alpha c_i}\right)^{1/\alpha}.
\end{align*}
A function $\chi_i$ is strictly increasing in $z_i$. Hence, the solution $\overline z_i(s)$ of (\ref{3.10}) satisfies the inequalities 
\begin{equation}  \label{3.11}
K\hat z_i(s)<\overline z_i(s)<\hat z_i(s),\quad s>s_i.
\end{equation}
Put
$$g(s)=K\sum_{i=1}^N\hat z_i(s)-s, \quad \overline\chi(s)=\sum_{i=1}^N \overline z_i(s)-s,\quad 
h(s)=\sum_{i=1}^N\hat z_i(s)-s.
$$
From (\ref{3.11}) we get
$$ g(s)<\overline\chi(s)<h(s),\qquad 
s>\max_{1\le i\le N} s_i.$$

For, $\hat s$ given by (\ref{3.9}), we have
\begin{align*}
h(\hat s)&=\sum_{i=1}^N\left(\frac{M}{\alpha c_i \hat s}\right)^{1/(\alpha-1)}-\hat s\\
&=\frac{1}{\hat s^{1/(\alpha-1)}}
\left(\left(\frac{M}{\alpha}\right)^{1/(\alpha-1)}\sum_{i=1}^N\frac{1}{c_i^{1/(\alpha-1)}}-\hat s^{\alpha/(\alpha-1)}\right)=0;\\
g(K^{(\alpha-1)/\alpha}\hat s)&=K\sum_{i=1}^N \hat z_i(K^{(\alpha-1)/\alpha}\hat s)-K^{(\alpha-1)/\alpha}\hat s\\
&=K^{(\alpha-1)/\alpha}\left(\sum_{i=1}^N \hat z_i(\hat s)-\hat s\right)=0.
\end{align*}
It follows that
\begin{equation}\label{3.12}
\overline\chi(K^{(\alpha-1)/\alpha}\hat s)>0,\qquad  \overline\chi(\hat s)<0,
\end{equation}
if $K^{(\alpha-1)/\alpha}\hat s>\max_{1\le i\le N} s_i.$ The last inequality reduces to
\begin{equation}\label{3.13}
 \sum_{j=1}^N \left(\frac{\min_{1\le i\le N} c_i}{c_j}\right)^{1/(\alpha-1)}\ge\frac{1}{1-K^{\alpha-1}}.
\end{equation} 
For any $K\in (0,1)$, by virtue of (\ref{3.7}), the inequality (\ref{3.13}), and, consequently, the inequalities (\ref{3.12}), hold true for $N$ large enough.  

From (\ref{3.12}) it follows that the unique solution $\overline s$ of (\ref{2.6F}), or, equivalently, of (\ref{2.6G}), satisfies the inequalities
\begin{equation} \label{3.14}
 K^{(\alpha-1)/\alpha}\hat s<\overline s<\hat s
\end{equation} 
for sufficiently large $N$, as far as $\overline\chi(s)/s=\sum_{i=1}^N\overline\sigma_i(s)-1$, and the share functions $\overline\sigma_i$ are strictly decreasing. Since $K\in (0,1)$ is arbitrary, we conclude that
\begin{equation} \label{3.15}
\overline s\sim \hat s,\quad N\to\infty.
\end{equation}

The assertion of the theorem follows from the relations (\ref{3.8}), (\ref{3.9}), (\ref{3.15}). 
\end{proof}

The limits (\ref{3.7A}) are similar to those of Theorem \ref{th:1}. Note that 
$$\lim_{N\to\infty}\mathscr A_N=\alpha^{1/\alpha}\le e^{1/e}\approx 1.445,\quad \alpha>1.$$
So, the asymptotic efficiency loss of the proportional compensation scheme, compared to the normative one, does not exceed $31\%$. 

\begin{remark} \label{rem:5}
In any compensation scheme under the assumption (\ref{3.7}) the total production tends to $+\infty$ as $N\to\infty$. The same is true in the setting of Theorem \ref{th:1} for $\alpha>1$. This is a consequence of the fact that the agent expenditures are very small for small amounts of output. In other words, they can start production almost for free. By hiring a large number of such agents the manager can ensure an arbitrary large output. 
\end{remark}

\begin{remark} \label{rem:6}
The total cost to total premium ratio
$$ D_N=\frac{1}{M}\sum_{i=1}^N c_i \overline z_i^\alpha(\overline s)$$ 
can be regarded as a measure of the reward dissipation of the proportional compensation scheme. For the normative scheme this ratio always equals to 1: see (\ref{2.8}). The estimates (\ref{3.11}), (\ref{3.14}) imply that 
$$ K \hat z_i(\hat s)< K\hat z_i(\overline s)<\overline z_i(\overline s)< \hat z_i(\overline s)<\hat z_i(K^{(\alpha-1)/\alpha}\hat s),$$
$$K^\alpha\sum_{i=1}^N c_i \hat z_i^\alpha(\hat s)<\sum_{i=1}^N c_i \overline z_i^\alpha(\overline s)<\sum_{i=1}^N c_i \hat z_i^\alpha(K^{(\alpha-1)/\alpha}\hat s)$$
for sufficiently large $N$. After the substitution of (\ref{3.10A}), (\ref{3.9}) we get
\begin{align*}
\frac{M}{\alpha}K^\alpha &=K^\alpha\sum_{i=1}^N c_i \left(\frac{M}{\alpha c_i\hat s}\right)^{\alpha/(\alpha-1)}<\sum_{i=1}^N c_i \overline z_i^\alpha(\overline s)\\
&<\frac{1}{K}\sum_{i=1}^N c_i \left(\frac{M}{\alpha c_i\hat s}\right)^{\alpha/(\alpha-1)}=\frac{M}{\alpha K}.
\end{align*}

Since $K\in (0,1)$ is arbitrary, it follows that $$\lim_{N\to\infty} D_N=1/\alpha>1.$$ Thus, the reward does not completely dissipate in the limit. This fact is also a consequence of the inequality $D_N\le(N-1)/(N\alpha)$, obtained in \cite[Theorem 5]{CorHar05}. 
\end{remark}

The condition (\ref{3.7}) is satisfied if the sequence $c_i$ is non-decreasing and
$$ \sum_{i=1}^\infty \frac{1}{c_i^{1/(\alpha-1)}}=+\infty.$$
It is not satisfied for $c_i^{1/(\alpha-1)}=cq^{i-1}$, $c>0$, $q>0$, $q\neq 1$. 

\begin{remark} \label{rem:7}
Assume that $c_i$ are independent identically distributed (i.i.d.) random variables, bounded from below by a positive constant: $c_i\ge \underline c>0$. Then the condition (\ref{3.7}) is satisfied almost surely:
$$ \sum_{i=1}^N \left(\frac{\min_{1\le k\le N} c_k}{c_i}\right)^{1/(\alpha-1)}\ge
   \underline c^{1/(\alpha-1)}\sum_{i=1}^N \left(\frac{1}{c_i}\right)^{1/(\alpha-1)}\to\infty\quad \textrm{a.s.},\quad N\to\infty.$$
Indeed, consider a sequence of strictly positive i.i.d. random variables $\xi_i$ and
take a constant $L>0$ such that $\nu=\mathsf E(\xi_i\wedge L)>0$, where $a\wedge b=\min\{a,b\}$. 
By the strong law of large numbers we have
$$ \liminf_{N\to\infty}\frac{1}{N}\sum_{i=1}^N\xi_i\ge  \lim_{N\to\infty}\frac{1}{N}\sum_{i=1}^N\xi_i\wedge L=\nu\quad \textrm{a.s.}$$
Hence, $\sum_{i=1}^\infty\xi_i\ge\sum_{i=1}^\infty\xi_i\wedge L=+\infty$.
\end{remark}   

The case of linear cost functions $\varphi_i(x)=c_i x$ appears to be more complex from the asymptotical point of view, although there is known an explicit expression for $\overline s$ in this case: see \cite[Proposition 5]{HilRil89}. For reader's convenience, in the next theorem we derive this expression using the argumentation similar to \cite[Theorem 4.19]{Voj16}. We only consider the price of anarchy $\mathscr A_N$, since from (\ref{2.4}), (\ref{2.5}) we see that $\mu=\max_{1\le i\le N} c_i$, and in this case optimal production levels $\hat x_i$ of reward-taking agents either equal to zero or are not uniquely defined. 
\begin{theorem} \label{th:3}
Assume that $\varphi_i(x)=c_i x$, $0<c_1\le\dots\le c_N$. Then
\begin{align}
\mathscr A_N &=\frac{1}{l-1}\sum_{i=1}^l\frac{c_i}{c_1}, \label{3.15A}\\
l &=\min\{i\in\{2,\dots,N-1\}:\frac{1}{i-1}\sum_{k=1}^i c_{k}\le c_{i+1}\}\wedge N,\nonumber
\end{align}
where $l$ is the number of active players (we use the convention $\min\emptyset=+\infty$). Moreover,  
\begin{equation} \label{3.15B}
 \frac{c_2}{c_1}<\mathscr A_N\le 2\frac{c_2}{c_1}.
\end{equation} 
\end{theorem}
\begin{proof} From (\ref{2.2}) by the Lagrange duality we get
$$ s^*=-\inf\left\{-\sum_{i=1}^N x_i: \sum_{i=1}^N c_i x_i\le M,\ x\ge 0\right\}=-\sup_{\lambda\ge 0}\theta(\lambda),$$
\begin{align*}
\theta(\lambda)&=\inf_{x\ge 0}\left\{-\sum_{i=1}^N x_i+\lambda\left(\sum_{i=1}^N c_i x_i-M\right)\right\}=-\lambda M+
\sum_{i=1}^N\inf_{x_i\ge 0} (\lambda c_i-1) x_i\\
&=\begin{cases}
-\lambda M,&\lambda\ge 1/c_1,\\
-\infty,& \lambda<1/c_1.
\end{cases}
\end{align*}
Thus, $s^*=M/c_1$. It also follows that $x_1^*=M/c_1$, $x_i^*=0$, $i\ge 2$ is an optimal solution of (\ref{2.2}). In other words, it is optimal to transfer total reward to an agent with the minimal marginal production cost $c_1$. 

Putting $\overline\sigma_i=\overline x_i/\overline s$, rewrite the equations (\ref{2.6A}) -- (\ref{2.6C}), determining the equilibrium of the proportional compensation scheme, as follows:
\begin{align}
\frac{M}{c_i}(1-\overline\sigma_i)&=\overline s,\quad \overline s<\frac{M}{c_i}, \label{3.16}\\
\overline x_i&=0,\quad  \overline s\ge\frac{M}{c_i},\nonumber\\
\sum_{i=1}^N\overline\sigma_i&=1,\quad \overline\sigma_i\ge 0,\quad \overline s>0\label{3.17}.
\end{align}
Denote by $l$ the number of active players. Then
\begin{equation} \label{3.18}
 l=\max\{i:\overline s<M/c_i\},\quad \overline s<M/c_l,\quad \overline s\ge M/{c_{l+1}}.
\end{equation} 
Using the equality (\ref{3.17}), from (\ref{3.16}) we get
\begin{equation} \label{3.19}
\overline s=M\frac{l-1}{\sum_{i=1}^l c_i}.
\end{equation}
As we know, there are at least two active players. Thus, from (\ref{3.18}), (\ref{3.19}) we conclude that $l$ can be expressed as follows
$$ l=\min\left\{i\in\{2,\dots,N-1\}:\frac{i-1}{\sum_{k=1}^i c_k}\ge\frac{1}{c_{i+1}}\right\}\wedge N.$$
This formula gives (\ref{3.15A}):
$$ \mathscr A_N=\frac{s^*}{\overline s}=\frac{1}{l-1}\sum_{i=1}^l\frac{c_i}{c_1}.$$

The inequalities of the form (\ref{3.15B}) are presented in \cite[Corollaries 4.21, 4.22]{Voj16}. The left inequality (\ref{3.15B}) is implied by (\ref{3.18}): $\overline s<M/c_l\le M/c_2.$
From (\ref{3.16}) it follows that $\overline\sigma_2\le\overline\sigma_1$. Hence, $\overline\sigma_2\le 1/2$ and
$$ \overline s=\frac{M}{c_2}(1-\overline\sigma_2)\ge \frac{M}{2 c_2}.$$
This gives the right inequality (\ref{3.15B}).
\end{proof}

One more estimate
\begin{equation} \label{3.22}
\mathscr A_N\le \frac{1}{N-1}\sum_{i=1}^N\frac{c_i}{c_1},\quad N\ge 2
\end{equation}
is obtained as follows. For $l=N$ this estimate turns into the equality. Furthermore, if $2\le l<N$, then
$ \frac{1}{l-1}\sum_{k=1}^l c_k\le c_{l+1}$, and (\ref{3.22}) is implied by the inequality 
\begin{align*}
\frac{1}{N-1}\sum_{k=1}^N c_k &=\frac{1}{N-1}\sum_{k=1}^l c_k + \frac{1}{N-1}\sum_{k=l+1}^N c_k\ge \frac{1}{N-1}\sum_{k=1}^l c_k + \frac{N-l}{N-1} c_{l+1}\\
&\ge  \frac{1}{N-1}\sum_{k=1}^l c_k + \frac{N-l}{N-1}\frac{1}{l-1}\sum_{k=1}^l c_k
=\frac{1}{l-1}\sum_{k=1}^l c_k.
\end{align*}

To obtain a meaningful asymptotic result, let us assume, as in Remark \ref{rem:7}, that $c_i$ are i.i.d. random variables 
and $c_i\ge\underline c>0$. Then, by the strong law of large numbers, from (\ref{3.22}) it follows that
$$ \lim_{N\to\infty}\mathscr A_N\le \frac{\nu}{\underline c}\quad \textrm{a.s.},$$
where $\nu=\mathsf Ec_i$. However, our numerical experiments suggest a much better result.
\begin{conjecture} \label{con:1} Assume that $\varphi_i(x)=c_i x$, where $c_i$ are i.i.d. random variables such that $c_i\ge \underline c>0$. Then
$$ \lim_{N\to\infty}\mathscr A_N=1\quad \textrm{a.s.}$$
\end{conjecture}

Denote by $c_k^{(N)}$ the $k$-th order statistics of the sequence $(c_i)_{i=1}^N$. That is, $c_k^{(N)}$ is a $k$-th smallest element of the sequence $(c_i)_{i=1}^N$:
$$ c_1^{(N)}\le\dots\le c_N^{(N)},\quad c_1^{(N)}=\min\{c_1,\dots,c_N\},\dots, c_N^{(N)}=\max\{c_1,\dots,c_N\}.$$
From (\ref{3.18}) it follows that
$$ \mathscr A_N=\frac{s^*}{\overline s}\le \frac{c_{l+1}^{(N)}}{c_1^{(N)}}.$$
To prove the Conjecture \ref{con:1} it would be enough to show that $c_{l+1}^{(N)}/{c_1^{N}}\to 1$ a.s., where
$$ l=\min\left\{i\in\{2,\dots,N-1\}:\frac{1}{i-1}\sum_{k=1}^i c_k^{(N)}\le c_{i+1}^{(N)}\right\}\wedge N.$$
But this is not an easy task.

The numerical experiments (implemented by the means of the $\mathsf{R}$ software), supporting this conjecture, are presented in Table \ref{tab:1}. As usual, by $U(a,b)$ we denote the uniform distribution on $(a,b)$. We write $\xi\sim LN(\mu,\sigma^2)$, if $\xi=e^{\mu+\sigma\eta}$, where $\eta$ is a standard normal random variable. The density of the Pareto distribution $\textrm{Pa}(\alpha,\lambda)$ is given by the formula
$$ f(x)=\frac{\alpha\lambda^\alpha}{(\lambda+x)^{\alpha+1}},\quad x>0.$$
Note that neither a heavy tail nor the infinite expectation (for $\alpha=0.5$) of the Pareto distribution prevent the convergence $\mathscr A_N\to 1$, $N\to\infty$.

\begin{table}
\caption{Sampled values of $\mathscr A_N-1$ for $c_i=1+\xi_i$, where the distributions i.i.d. random variables $\xi_i$ are indicated in the first row.}
\label{tab:1}       
\begin{tabular}{lllllll}
\hline\noalign{\smallskip}\\
$N$ &  $U(1,2)$ &     $U(1,10)$  &      $LN(0,1)$      &  $LN(0,2)$          & $\textrm{Pa}(0.5,1)$ &  $\textrm{Pa}(3,1)$\\
\noalign{\smallskip}\hline\noalign{\smallskip}
$10^2$  & $1.3\cdot 10^{-1}$  & $5.2\cdot 10^{-1}$ &  $2.6\cdot 10^{-1}$ &  $1.1\cdot 10^{-1}$ & $2.4\cdot 10^{-1}$ & $9.4\cdot 10^{-2}$\\
$10^3$  & $4.3\cdot 10^{-2}$  & $1.4\cdot 10^{-1}$ &  $9.6\cdot 10^{-2}$ &  $4.1\cdot 10^{-2}$ & $5.0\cdot 10^{-2}$ &  $2.7\cdot 10^{-2}$\\
$10^4$  & $1.4\cdot 10^{-2}$  & $5.1\cdot 10^{-2}$ &  $5.0\cdot 10^{-2}$ &    $1.4\cdot 10^{-2}$ & $2.0\cdot 10^{-2}$ &  $8.3\cdot 10^{-3}$\\
$10^5$  & $4.3\cdot 10^{-3}$  & $1.4\cdot 10^{-2}$ &  $3.5\cdot 10^{-2}$   &$5.3\cdot 10^{-3}$ & $6.5\cdot 10^{-3}$ & $2.5\cdot 10^{-3}$\\
$10^6$  & $1.4\cdot 10^{-3}$  & $4.3\cdot 10^{-3}$ &  $2.2\cdot 10^{-2}$ &  $2.2\cdot 10^{-3}$ & $2.0\cdot 10^{-3}$&
$8.3\cdot 10^{-4}$\\
$10^7$  & $4.4\cdot 10^{-4}$  & $1.3\cdot 10^{-3}$ &  $1.2\cdot 10^{-2}$ &  $1.0\cdot 10^{-3}$   & $6.3\cdot 10^{-4}$&
$2.6\cdot 10^{-4}$
\\
$10^8$  & $1.4\cdot 10^{-4}$  & $4.2\cdot 10^{-4}$ &  $9.3\cdot 10^{-3}$ &  $4.7\cdot 10^{-4}$ & $2.0\cdot 10^{-4}$&
$8.1\cdot 10^{-5}$
\\
\noalign{\smallskip}\hline
\end{tabular}
\end{table}

Table 2 indicates that under the assumptions of Conjecture \ref{con:1} the proportion of active players $l/N$ tends to zero. 
\begin{table}
\caption{Sampled values of the proportion $l/N$ of active players for $c_i=1+\xi_i$, where the distributions i.i.d. random variables $\xi_i$ are indicated in the first row.}
\label{tab:2}       
\begin{tabular}{lllllll}
\hline\noalign{\smallskip}\\
$N$ &  $U(1,2)$ &     $U(1,10)$  &      $LN(0,1)$      &                $LN(0,2)$          & $\textrm{Pa}(0.5,1)$ &  $\textrm{Pa}(3,1)$\\
\noalign{\smallskip}\hline\noalign{\smallskip}
$10^2$  & $1.7\cdot 10^{-1}$  &  $5.0\cdot 10^{-2}$ & $1.1\cdot 10^{-1}$ & $2.3\cdot 10^{-1}$ & $8.0\cdot 10^{-2}$ & $2.2\cdot 10^{-1}$\\
$10^3$  & $4.4\cdot 10^{-2}$  &  $1.7\cdot 10^{-2}$ & $3.0\cdot 10^{-2}$ & $4.6\cdot 10^{-2}$ & $3.6\cdot 10^{-2}$ & $6.8\cdot 10^{-2}$\\
$10^4$  & $1.4\cdot 10^{-2}$  &  $4.3\cdot 10^{-3}$ & $6.4\cdot 10^{-3}$ & $1.7\cdot 10^{-2}$ & $1.0\cdot 10^{-2}$ & $2.3\cdot 10^{-2}$\\
$10^5$  & $4.7\cdot 10^{-3}$  &  $1.5\cdot 10^{-3}$ & $1.0\cdot 10^{-3}$ & $4.5\cdot 10^{-3}$ & $3.1\cdot 10^{-3}$ & $7.9\cdot 10^{-3}$\\
$10^6$  & $1.4\cdot 10^{-3}$ &   $4.8\cdot 10^{-4}$ & $1.9\cdot 10^{-4}$ & $1.2\cdot 10^{-3}$ & $9.8\cdot 10^{-4}$ & $2.5\cdot 10^{-3}$\\
$10^7$  & $4.5\cdot 10^{-4}$ &   $1.5\cdot 10^{-4}$ & $3.2\cdot 10^{-5}$ & $2.8\cdot 10^{-4}$ & $3.2\cdot 10^{-4}$ & $7.7\cdot 10^{-4}$\\
$10^8$  & $1.4\cdot 10^{-4}$ &   $4.7\cdot 10^{-5}$ & $5.0\cdot 10^{-6}$ & $6.5\cdot 10^{-5}$ & $1.0\cdot 10^{-4}$ & $2.5\cdot 10^{-4}$\\
\noalign{\smallskip}\hline
\end{tabular}
\end{table}

\begin{remark}
Concluding the paper, we recall a popular concept of \emph{no-regret learning}, explaining the emergence of an equilibrium in a repeated game. Consider a game with the payoff functions 
$u_i(x_1,\dots,x_N)$, $i=1,\dots,N$, defined on $S_1\times\dots\times S_N$, and assume that an agent $i$ knows only his own payoff $u_i$ and picks his strategy $x^t_i\in S_i$ according to a \emph{no-regret algorithm}, ensuring that
\begin{equation} \label{3.23}
\max_{y\in S_i}\sum_{t=1}^T u_i(y,x_{-i}^t)-\sum_{t=1}^T u_i(x^t_i,x_{-i}^t)=o(T),\quad T\to\infty,  
\end{equation} 
where $x_{-i}^t=(x_k^t)_{k\neq i}$. Under appropriate assumptions, a plenty of such algorithms is provided by the theory of online convex optimization: see, e.g., \cite{Haz16}. 

In \cite{EveManNad09} for the class of \emph{socially concave} games it was proved that the average strategy vector $\frac{1}{T}\sum_{t=1}^T x^t$ converges to a Nash equilibrium, if $x^t$ satisfies the no-regret property (\ref{3.23}). The proportional compensation scheme is a socially concave game under a technical assumption $x_i\in S_i=[b_{\min},b_{\max}]\subset (0,\infty)$: see \cite[Lemma 5.4]{EveManNad09}. Thus, at least a positive Nash equilibrium is approximated by a no-regret dynamics. 
\end{remark}

\bibliographystyle{spmpsci}
\bibliography{litCompens}

\begin{thebibliography}{10}
\providecommand{\url}[1]{{#1}}
\providecommand{\urlprefix}{URL }
\expandafter\ifx\csname urlstyle\endcsname\relax
  \providecommand{\doi}[1]{DOI~\discretionary{}{}{}#1}\else
  \providecommand{\doi}{DOI~\discretionary{}{}{}\begingroup
  \urlstyle{rm}\Url}\fi

\bibitem{BoyVan04}
Boyd, S., Vandenberghe, L.: Convex optimization.
\newblock Cambridge University Press, New York (2004)

\bibitem{CarVou16}
Caragiannis, I., Voudouris, A.: Welfare guarantees for proportional
  allocations.
\newblock Theory of Computing Systems \textbf{59}(4), 581--599 (2016)

\bibitem{CasMasShe10}
Cason, T., Masters, W., Sheremeta, R.: Entry into winner-take-all and
  proportional-prize contests: an experimental study.
\newblock Journal of Public Economics \textbf{94}(9), 604--611 (2010)

\bibitem{Che12}
Chen, B.Y.: Classification of $h$-homogeneous production functions with
  constant elasticity of substitution.
\newblock Tamkang Journal of Mathematics \textbf{43}(2), 321--328 (2012)

\bibitem{ChriSgoTan16}
Christodoulou, G., Sgouritsa, A., Tang, B.: On the efficiency of the
  proportional allocation mechanism for divisible resources.
\newblock Theory of Computing Systems \textbf{59}(4), 600--618 (2016)

\bibitem{Con14}
Connelly, B., Tihanyi, L., Crook, T., Gangloff, K.: Tournament theory thirty
  years of contests and competitions.
\newblock Journal of Management \textbf{40}(1), 16--47 (2014)

\bibitem{Cor07}
Corch{\'o}n, L.: The theory of contests: a survey.
\newblock Review of Economic Design \textbf{11}(2), 69--100 (2007)

\bibitem{CorSer16}
Corch{\'o}n, L., Serena, M.: Contest theory: a survey.
\newblock Handbook of game theory and industrial organization, Forthcoming,
  http://dx.doi.org/10.2139/ssrn.2811686

\bibitem{CorHar05}
Cornes, R., Hartley, R.: Asymmetric contests with general technologies.
\newblock Economic Theory \textbf{26}(4), 923–--946 (2005)

\bibitem{Cortak16}
Cornes, R., Sato, T.: Existence and uniqueness of {N}ash equilibrium in
  aggregative games: an expository treatment, pp. 47--61.
\newblock Springer International Publishing, Cham (2016)

\bibitem{CotMil99}
Coto-Mill{\'a}n, P.: Utility and production.
\newblock Physica-Verlag, Heidelberg (1999)

\bibitem{DecKovShe15}
Dechenaux, E., Kovenock, D., Sheremeta, R.: A survey of experimental research
  on contests, all-pay auctions and tournaments.
\newblock Experimental Economics \textbf{18}(4), 609--669 (2015)

\bibitem{EveManNad09}
Even-Dar, E., Mansour, Y., Nadav, U.: On the convergence of regret minimization
  dynamics in concave games.
\newblock In: Proceedings of the 41st ACM symposium on theory of computing, pp.
  523--532. ACM (2009)

\bibitem{GreSto83}
Green, J., Stokey, N.: A comparison of tournaments and contracts.
\newblock Journal of Political Economy \textbf{91}(3), 349--364 (1983)

\bibitem{Guth16}
G{\"u}th, W., Lev{\'\i}nsk{\`y}, R., Pull, K., Weisel, O.: Tournaments and
  piece rates revisited: a theoretical and experimental study of
  output-dependent prize tournaments.
\newblock Review of Economic Design \textbf{20}(1), 69--88 (2016)

\bibitem{Haz16}
Hazan, E.: Introduction to online convex optimization.
\newblock Foundations and Trends in Optimization \textbf{2}(3-4), 157--325
  (2016)

\bibitem{HilRil89}
Hillman, A., Riley, J.: Politically contestable rents and transfers.
\newblock Economics \& Politics \textbf{1}(1), 17--39 (1989)

\bibitem{JiaSkaVai13}
Jia, H., Skaperdas, S., Vaidya, S.: Contest functions theoretical foundations
  and issues in estimation.
\newblock International Journal of Industrial Organization \textbf{31}(3),
  211--222 (2013)

\bibitem{JohTsi04}
Johari, R., Tsitsiklis, J.: Efficiency loss in a network resource allocation
  game.
\newblock Mathematics of Operations Research \textbf{29}(3), 407--435 (2004)

\bibitem{Kel97}
Kelly, F.: Charging and rate control for elastic traffic.
\newblock European Transactions on Telecommunications \textbf{8}(1), 33--37
  (1997)

\bibitem{Kon09}
Konrad, K.: Strategy and dynamics in contests.
\newblock Oxford University Press, New York (2009)

\bibitem{KouPap09}
Koutsoupias, E., Papadimitriou, C.: Worst-case equilibria.
\newblock Computer Science Review \textbf{3}(2), 65--69 (2009)

\bibitem{LazRos81}
Lazear, E., Rosen, S.: Rank-order tournaments as optimum labor contracts.
\newblock Journal of Political Economy \textbf{89}(5), 841--864 (1981)

\bibitem{Mal86}
Malcomson, J.: Rank-order contracts for a principal with many agents.
\newblock The Review of Economic Studies \textbf{53}(5), 807--817 (1986)

\bibitem{MouQua16}
von Mouche, P., Quartieri, F.: Equilibrium theory for {C}ournot oligopolies and
  related games.
\newblock Springer International Publishing Switzerland (2016)

\bibitem{NalSti83}
Nalebuff, B., Stiglitz, J.: Towards a general theory of compensation and
  competition.
\newblock The Bell Journal of Economics \textbf{14}(1), 21--43 (1983)

\bibitem{OkuSzi99}
Okuguchi, K., Szidarovszky, F.: The theory of oligopoly with multi-product
  firms.
\newblock Springer, Berlin (1999)

\bibitem{Pat08}
Patriksson, M.: A survey on the continuous nonlinear resource allocation
  problem.
\newblock European Journal of Operational Research \textbf{185}(1), 1--46
  (2008)

\bibitem{She16}
Sheremeta, R.: The pros and cons of workplace tournaments.
\newblock Working Paper 16-27, Chapman University, Economic Science Institute
  (2016)

\bibitem{SziOku97}
Szidarovszky, F., Okuguchi, K.: On the existence and uniqueness of pure {N}ash
  equilibrium in rent-seeking games.
\newblock Games and Economic Behavior \textbf{18}(1), 135--140 (1997)

\bibitem{VilVil17}
V{\^\i}lcu, A.D., V{\^\i}lcu, G.E.: A survey on the geometry of production
  models in economics.
\newblock Arab Journal of Mathematical Sciences \textbf{23}(1), 18--31 (2017)

\bibitem{Voj16}
Vojnovi{\'c}, M.: Contest theory incentive mechanisms and ranking methods.
\newblock Cambridge University Press, New York (2016)

\end{thebibliography}
\end{document}